\pdfoutput=1
\documentclass[10pt,pra,aps,twocolumn]{revtex4-1}
\usepackage{amsmath,bbm}
\usepackage{amsthm}
\usepackage{amsmath}
\usepackage{latexsym}
\usepackage{amssymb}
\usepackage{float}
\usepackage{graphicx}           
\usepackage{color}
\usepackage{mathpazo}
\usepackage{comment}
\usepackage[colorlinks=true,linkcolor=blue,citecolor=magenta,urlcolor=blue]{hyperref}

\newcommand{\be}{\begin{equation}}
\newcommand{\ee}{\end{equation}}
\newcommand{\bea}{\begin{eqnarray}}
\newcommand{\eea}{\end{eqnarray}}

\def\squareforqed{\hbox{\rlap{$\sqcap$}$\sqcup$}}
\def\qed{\ifmmode\squareforqed\else{\unskip\nobreak\hfil
\penalty50\hskip1em\null\nobreak\hfil\squareforqed
\parfillskip=0pt\finalhyphendemerits=0\endgraf}\fi}
\def\endenv{\ifmmode\;\else{\unskip\nobreak\hfil
\penalty50\hskip1em\null\nobreak\hfil\;
\parfillskip=0pt\finalhyphendemerits=0\endgraf}\fi}

\newcommand{\tr}{\text{Tr}}
\newcommand{\I}{\mathbbm{1}}
\newcommand{\B}{\mathcal{B}}

\newcommand{\D}{\mathcal{D}}
\newcommand{\w}{\omega}
\newcommand{\ket}[1]{|#1\rangle}

\newcommand{\la}{\langle}
\newcommand{\ra}{\rangle}

\makeatletter
\newtheorem*{rep@theorem}{\rep@title}
\newcommand{\newreptheorem}[2]{%
\newenvironment{rep#1}[1]{%
 \def\rep@title{#2 \ref{##1}}%
 \begin{rep@theorem}}%
 {\end{rep@theorem}}}
\makeatother

\newtheorem{thm}{Theorem}%[section]
\newreptheorem{thm}{Theorem}

\newtheorem{fact}{Fact}

\newtheorem{proposition}{Proposition}

\begin{document}

\title{Perfect discrimination of quantum measurements using entangled systems}

\author{Chandan Datta}
\affiliation{Centre for Quantum Optical Technologies, Centre of New Technologies, University of Warsaw, Banacha 2c, 02-097 Warszawa, Poland}

\author{Tanmoy Biswas}
\affiliation{International  Centre  for  Theory  of  Quantum  Technologies,  University  of  Gdansk, Wita Stwosza 63, 80-308 Gda\'nsk,  Poland}
\affiliation{Institute of Theoretical Physics and Astrophysics, National Quantum
Information Centre, Faculty of Mathematics, Physics and Informatics,
University of Gdansk, Wita Stwosza 57, 80-308 Gda\'nsk, Poland}

\author{Debashis Saha}
\email{saha@cft.edu.pl}
\affiliation{Center for Theoretical Physics, Polish Academy of Sciences, Aleja Lotnik\'{o}w 32/46, 02-668 Warsaw, Poland}

\author{Remigiusz Augusiak}
\email{augusiak@cft.edu.pl}
\affiliation{Center for Theoretical Physics, Polish Academy of Sciences, Aleja Lotnik\'{o}w 32/46, 02-668 Warszawa, Poland}

\begin{abstract}
Distinguishing physical processes is one of the fundamental problems in quantum physics. Although distinguishability of quantum preparations and quantum channels have been studied considerably, distinguishability of quantum measurements remains largely unexplored. We investigate the problem of single-shot discrimination of quantum measurements using two strategies, one based on single quantum systems and the other one based on entangled quantum systems. First, we formally define both scenarios. We then construct sets of measurements (including non-projective) in arbitrary finite dimensions that are perfectly distinguishable within the second scenario using quantum entanglement, while not in the one based on single quantum systems. Furthermore, we show that any advantage in measurement discrimination tasks over single systems is a demonstration of Einstein–Podolsky–Rosen ‘quantum steering’. Alongside, we prove that all pure two-qubit entangled states provide an advantage in a measurement discrimination task over one-qubit systems.

\end{abstract}

\maketitle

%\onecolumngrid 

\section{Introduction}

Distinguishability of different physical processes is a fundamental question in the field of quantum physics \cite{cs2019quantum}.  It all started with the seminal work on quantum state discrimination by Helstrom \cite{helstrom1976}, in which an upper bound on the probability of discrimination between the two states was derived. In recent years the problem of state discrimination has been explored extensively not only from the fundamental perspective but has also been studied due to its relevance for quantum information protocols such as quantum communication or quantum cryptography \cite{chefles2000, paris2004, bergou2007, bae2015}. 
It has also been studied in the context of resource theories of measurements
in the quantum mechanical \cite{OszBisl,UolaPRL} as well as the generalized probability theory \cite{TakPRX} scenario. On the other hand, via the well-known Choi-Jamio\l{}kowski isomorphism \cite{jamiolkowski1972, choi1975}, the state discrimination problem of Helstrom has been, quite naturally, translated to the problem of discrimination of quantum channels, whose various aspects have intensively been studied in recent years \cite{acin2001, sacchi2005, piani2005, pianiPRA, duan2009, harrow2010, ziman2010, piani2015}.
 
At the same time, the problem of quantum measurement discrimination remains relatively unexplored. It should be noted here that the most intriguing feature of the measurement discrimination problem, as compared to state discrimination, is that in it one can enhance the probability of distinguishing measurements by using quantum entanglement \cite{ji2006, ziman, puchala2018}. In \cite{ziman2008, ziman2009} unambiguous discrimination of quantum measurements is reported, where two shots are needed for perfect discrimination. Furthermore, in a recent experiment \cite{fiurasek2009} the optimal discrimination of two projective quantum measurements has been investigated. However, most of the previous results, are either restricted to two-dimensional quantum systems \cite{ji2006, ziman, ziman2008, fiurasek2009}, or have been studied in the multiple-shot scenario \cite{ji2006, OszMul}. Like in \cite{ziman}, here we consider the most practical single-shot scenario wherein only one copy of the measurement device is available, and the measurements are destructive, that is, we assume that there is no access to the post-measurement state. 

The main aim of this work is to explore the advantage in discriminating arbitrary dimensional measurements provided by entangled quantum systems. Precisely, we derive criteria allowing to decide whether a set of measurements in arbitrary dimension $d$ is perfectly distinguishable in entanglement-assisted scenario, but not within the single-system scenario for: (i) rank-one projective measurements, (ii) $d^2$ outcome positive operator-valued measurements (POVMs), and (iii) $d+1$ outcome POVMs. Moreover, we provide many classes of measurement of these types and further study the qualitative advantage obtained in entanglement-assisted scenario as compared to the single-system scenario. Remarkably, any advantage provided by the entangled systems in the measurement discrimination problem is a proof of `quantum steering'. Quantum steering  manifests the nonlocal effects of entangled states, which was scrutinized by Einstein-Podolsky-Rosen \cite{epr} and formalized later in \cite{wiseman}. Alongside, we show there exists a set of measurements (up to a local unitary equivalence) for which all pure two-qubit entangled states provide advantage. 

The paper is organised as follows. In Sec. \ref{framework}, we formulate the measurement distinguishability task in two different scenarios and identify the necessary resource for the advantage with entangled systems. The advantage of the entangled systems over the single systems is studied in Sec. \ref{advantage}. Finally, we conclude in Sec. \ref{conclusion} with many possibilities of future investigation.

\section{Framework for measurement distinguishability problem}\label{framework}

In this section, we formulate the quantum measurement distinguishability problem and then describe two scenarios in which it can be addressed.

Assume that we are given a measurement device that performs one of $n$ a priori known $m-$outcome measurements $M_x := \{M^a_{x}\}_a$, where $M^a_{x}$ stands for a measurement operator corresponding to the outcome $a$ of the $x$-th measurement with $x \in [n] = \{0,\dots,n-1\}$ and $a=[m]=\{0,\ldots,m-1\}$. These measurements are sampled from the probability distribution $p(x)$. We additionally assume that all the measurements act on $\mathbbm{C}^d$ where, in general, $d$ may not be equal to $m$.   
%For instance, in the simplest case the measurement device can perform one of the two measurements with equal probabilities. 
In order to distinguish the measurements, the measurement device is fed with a known quantum state that belongs to $\mathbbm{C}^d$ and the device performs one of the measurements $M_x$ with probability $p(x)$ on it. We assume that 
there is no access to the post-measurement state. The single-shot measurement distinguishability problem consists in maximizing the probability of correctly guessing which measurement has been performed based solely on the obtained outcome $a$.

%The user has full control over the preparation device from which a known quantum %system is fed to the given measurement device, but there is no access to the %post-measurement quantum system. 

Now, we formalize two different scenarios to address the above task: (i) the one in which the measurements are performed on a single quantum system (which in principle can be entangled to another quantum system, however, we do not have access to it), and (ii) the one exploiting quantum entanglement in which the particle going through our measurement device is quantum mechanically correlated to another system which we 
have access to and can measure it.

\subsection{Measurement distinguishability with single systems} 

To discriminate the measurement we consider the following strategy using single quantum system (Fig. \ref{fig single system}). Given a known quantum preparation $\rho$, the measurement device performs a measurement $M_x$ on it. Upon obtaining the outcome $a$ we can, in general, perform a classical post processing defined by $Q:=\{q(z|a)\}$ where $z\in \{0,\dots,n-1\}$ and $\forall a,\ \sum_z q(z|a)=1$. Post-processing is just a stochastic map that acts on the probability distribution $p(a|M_x,\rho)$ and returns $\bar{p}(z|M_x,\rho)$ as the output $z$ (the guess) is supposed to be equal to $x$. The optimal single-shot distinguishing probability of the measurement set, denoted by $\D$, is given by
\bea 
\D &=& \max_{\rho,\{q(z|a)\}} \sum_{x} p(x) \bar{p}(z=x|M_x,\rho) \nonumber \\
&=& \max_{\rho,\{q(z|a)\}} \sum_{x,a} q(z=x|a) p(x) \mbox{Tr}(\rho M^a_x) .
\eea  
% \bea\D&=& \max_{\rho,\{q(z|a)\}} \sum_{a}  \big\{ q(0|a) p(0) \mbox{Tr}(\rho M^a_0) + \dots \nonumber \\  && \qquad \qquad \quad + q(n-1|a) p(n-1) \mbox{Tr}(\rho M^a_{n-1}) \big\} \nonumber \\&=& \max_{\rho} \sum_{a} \max \big\{ p(0) \mbox{Tr}(\rho M^a_0) , \dots , p(n-1) \mbox{Tr}(\rho M^a_{n-1}) \big\}  . \nonumber \\\eea 
We note that the term inside the summation is just a convex mixture of $p(x) \mbox{Tr}(\rho M^a_x)$ with weightage $q(z|a)$ for every $a$, and thus we choose the appropriate $\{q(z|a)\}$ which yields the optimal distinguishing probability. Therefore, the above quantity 
 can be expressed only in terms of $\mbox{Tr}(\rho M^a_x)$ as follows
\bea \label{Drho}
\D \label{D}
&= & \max_{\rho} \sum_a \max_x \left\{ p(x) \mbox{Tr}(\rho M^a_x) \right\} .
\eea 
Note that it is sufficient to consider only pure quantum states to obtain $\D $. This is because the $\max_x \{p(x)\mathrm{Tr}(\rho M^x_a)\}$ is a convex function of the state $\rho$. Following this argument we can infer that the maximum is always achieved by a pure quantum state. Let us finally mention that $\D=1$ means that the corresponding strategy perfectly distinguishes the measurements.

\begin{figure}[http]
\centering
\includegraphics[scale=0.25]{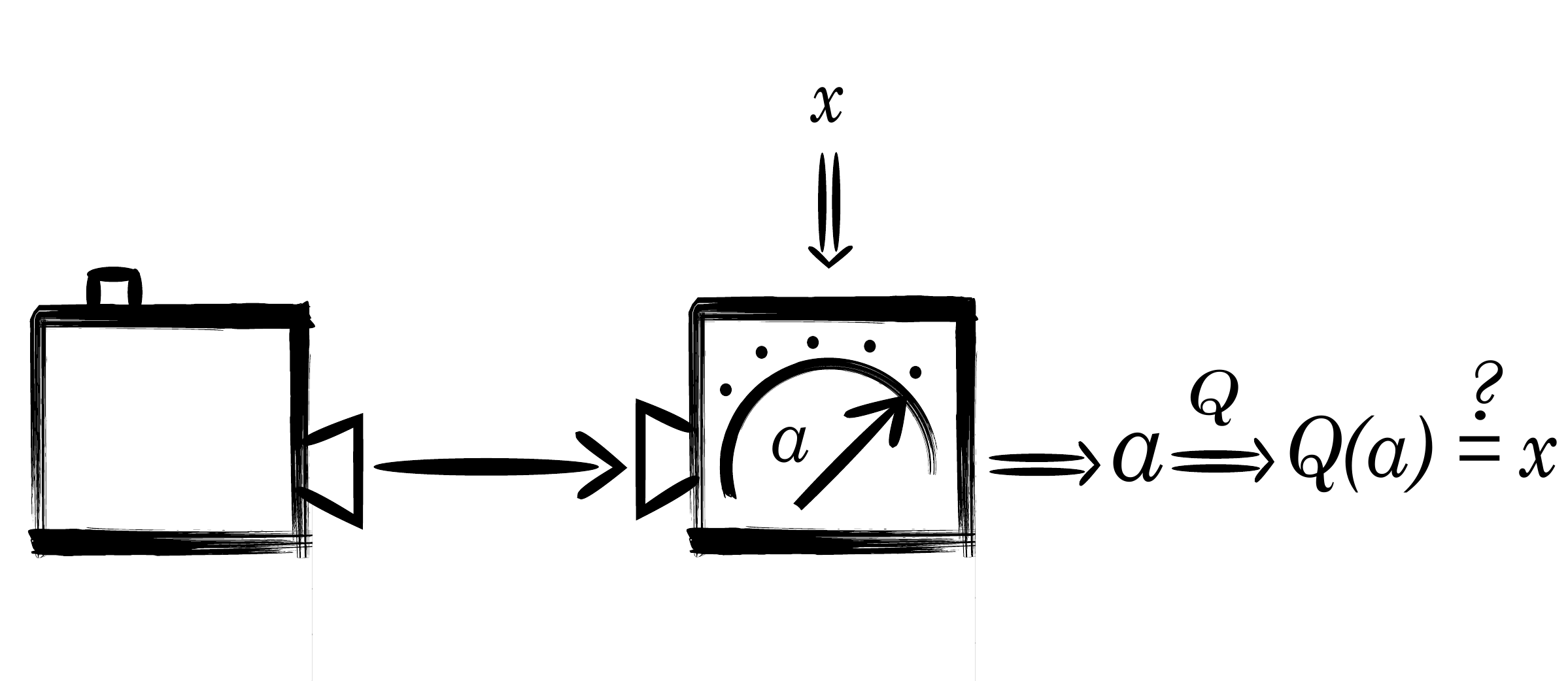}
\caption{\textbf{Schematic illustration of the single-shot measurement distinguishability problem in the single-system scenario}. A measurement device is fed with a known state $\rho$ and performs a measurement $M_x$ on it, yielding an outcome $a$. Depending on the outcome one then chooses the post-processing strategy to provide the best guess for the input $x$. The user has full control over the preparation device.}
\label{fig single system}
\end{figure}
 
\begin{figure}[http]
\centering
\includegraphics[scale=0.35]{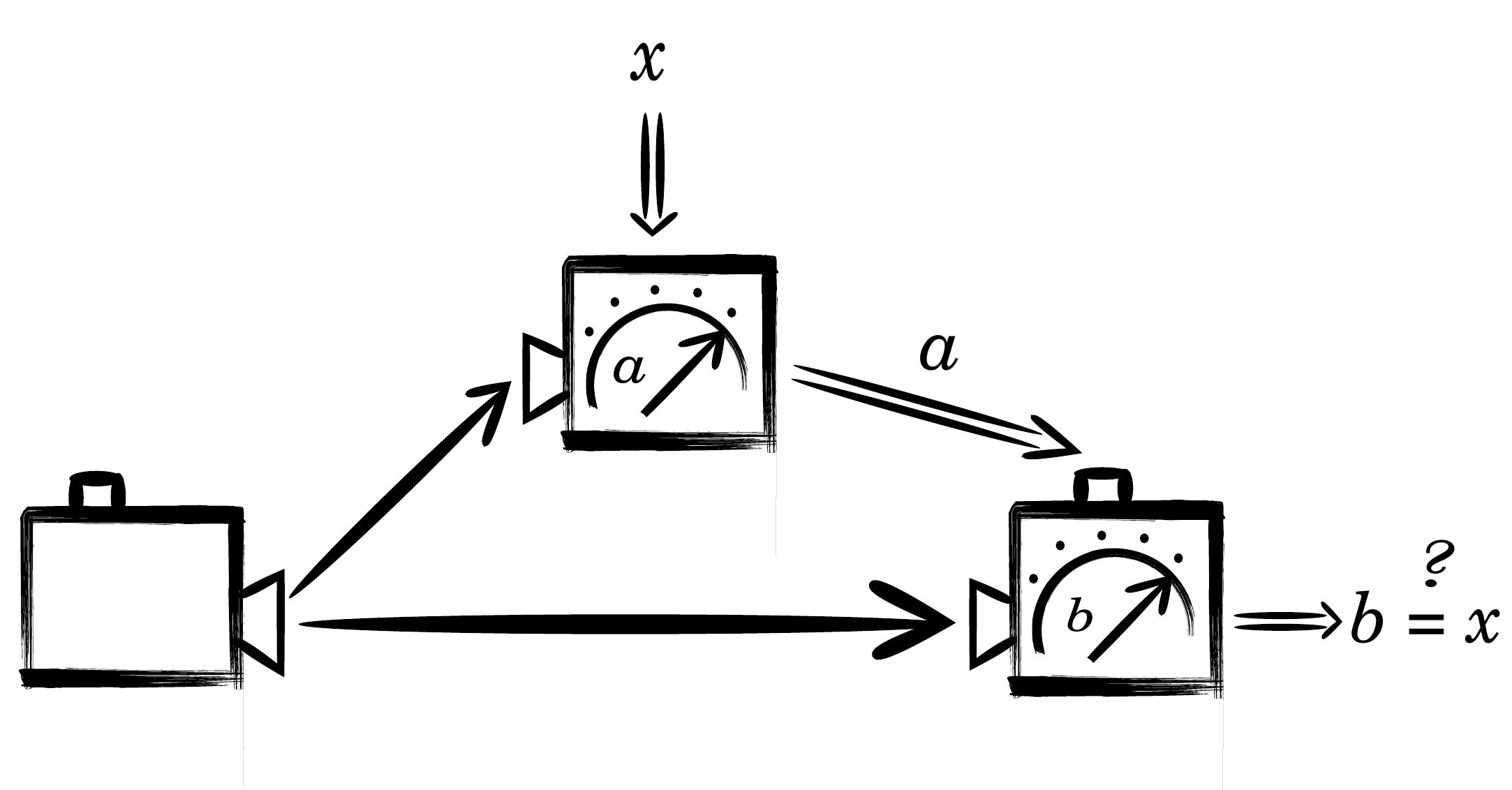}
\caption{\textbf{Schematic illustration of the entanglement-assisted measurement distinguishability problem.} A known bipartite state $\rho_{AB}$ is sent to two measurement devices held by Alice and Bob. However, there is a catch. Alice performs one of the measurement from the set $\{M_x\}$ and obtains a outcome $a$. Depending on that outcome Bob selects a measurement $N_y$ from the known set $\{N_y\}$, such that the outcome $b$ guesses the measurement setting $x$ of Alice. Devices with nob are fully controllable by the users.} 
\label{fig entanglement assisted}
\end{figure}

\subsection{Entanglement-assisted measurement distinguishability} 

As shown in Fig. \ref{fig entanglement assisted} the strategy is as follows. A known bipartite state $\rho_{AB}$ is sent to Alice and Bob. Alice performs the unknown measurement $M_x$ from the set $\{M_x\}$ on her sub-system and sends the outcome $a$ to Bob. Upon obtaining the outcome from Alice, Bob chooses to perform one of $m$ measurements, denoted $N_y$, which yields an outcome $b$. Bob's outcome is his guess of Alice's input. It is noteworthy that any classical post-processing of the outcome $a$ can be included into Bob's measurement. 

We can express the distinguishing probability in such entanglement-assisted scenario by a Bell expression.
Precisely, the experiment made by Alice and Bob can be described by the joint probabilities $p(a,b|x,y)$, where $x,y$ and $a,b$ are input and output variables, respectively. A general linear Bell expression in this scenario is
\begin{equation}\label{BellExp}
\sum_{x,y,a,b} c_{x,y,a,b} p(a,b|x,y),
\end{equation}
where $c_{x,y,a,b}$ are some real coefficients.

In the context of measurement distinguishability task, Alice's device performs $n$ $m$-outcome measurements, whereas Bob can choose to perform one of $m$ measurements, each having $n$ outcomes; so $x,b\in\{0,\ldots,n-1\}$ and $y,a\in \{0,\ldots,m-1\}$. Moreover, we are interested in only those cases where Bob's input $y$ is same as Alice's output $a$, and Bob's output $b$ should be the guess of Alice's input $x$. Thus, the distinguishing probability in the entanglement-assisted scenario is given by 
\be 
\sum_{x,a} p(x) p(a,b=x|x,y=a), 
\ee
and it is a particular instance of the Bell expression (\ref{BellExp}) with $c_{x,y,a,b} = p(x)$ when $y= a,b= x$, and $c_{x,y,a,b}=0$ otherwise.  In quantum theory, the entanglement-assisted distinguishing probability pertaining to a particular shared state $\rho_{AB}$ is expressed as
\be \label{qbxs}
\B^{\ \rho_{AB}} = \max_{\{N^b_y\}} \ \sum_{x,a} p(x) \mbox{Tr} \left[\rho_{AB} \left(M^a_x \otimes N^{b=x}_{y=a}\right)\right],
\ee 
where $\{M_x\}$ is the given set of quantum measurements (with probability $p(x)$) which we want to distinguish. 
The optimal entanglement-assisted distinguishing probability (denoted by $\B$) can be achieved by optimizing over all the bipartite states, that is,
\be \label{qbx}
\B = \max_{\rho_{AB}} \B^{\ \rho_{AB}} .
\ee
Analogously to the previous case, here the maximum is also achieved by a pure state. Note that the maximum value of $\B=1$ corresponds to the case of perfect distinguishability. In the following subsection we show that the advantage in a measurement distinguishability problem for a entangled bipartite system over a single system can be related to a necessary condition for quantum steering.

\subsection{Quantum steering is the necessary condition for advantage in a measurement distinguishability task}\label{necessary steering}

Let us now make a connection between the measurement distinguishability task and
quantum steering.
\begin{thm}\label{thm1}
Given any set of measurements $\{M_x\}$, $\B^{\ \rho_{AB}} > \D$ implies the shared state $\rho_{AB}$ is steerable by Bob.
\end{thm}
\begin{proof}
If the shared state $\rho_{AB}$ has a local hidden-state (LHS) model for any choice of measurements in a given scenario where Alice has the quantum device, then for all $a,b,x,y,$ the joint probabilities obtained from that state can be written as \cite{wiseman},
\be \label{plhs}
p(a,b|x,y) = \sum_\lambda \mu(\lambda)  \mbox{Tr}(\rho_\lambda M^a_x) p_B(b|y,\lambda),
\ee 
where $\sum_\lambda\mu(\lambda)=1$, $p_B(b|y,\lambda)$ represents arbitrary probability distributions depending on the hidden variable $\lambda$, which, without any loss of generality, can always be assumed to be deterministic. Finally, $\mbox{Tr}(\rho_\lambda M^a_x)$ corresponds to the probability of obtaining the outcome $a$ when the measurement $x$ is performed on the hidden state $\rho_\lambda$. Whenever the joint probabilities cannot be expressed in the above form, then the shared state is steerable from Bob to Alice. Now, let us define the distinguishing probability in entanglement-assisted scenario when the shared state admits a LHS model, 
\be 
\B' = \max_{\rho_{AB}\in \mathcal{LHS}} \ \sum_{x,a} p(x) p(a,b=x|x,y=a),
\ee
wherein $\mathcal{LHS}$ denotes the set of states admitting LHS model so that the joint probabilities are given by Eq. \eqref{plhs}.
To prove the desired result, it suffices to show $\B' = \D$. 
Notice that by using Eq. \eqref{plhs} we obtain
\bea  
&&\sum_{a,x} p(x) \sum_\lambda \mu(\lambda) \mathrm{Tr}(\rho_\lambda M^a_x) p_B(b=x|y=a,\lambda) \nonumber \\
&=& \sum_{a}\sum_\lambda \mu(\lambda) \left( \sum_x p(x)  \mathrm{Tr}(\rho_\lambda M^a_x) p_B(b=x|y=a,\lambda) \right) \nonumber \\
&\leqslant &  \sum_a \sum_\lambda \mu(\lambda) \max_x \left\{ p(x) \mathrm{Tr}(\rho_\lambda M^a_x) \right\}, 
%\nonumber \\&&\hspace{0.5cm}=  \sum_a \max_x \left\{ p(x) \mathrm{Tr}(\rho_A M^a_x) \right\}, 
\eea  
where the inequality stems from the fact that $p_{B}(b|y,\lambda)$ is a deterministic probability distribution for any choice of the measurement $y$ and the variable $\lambda$. The above clearly implies that
\begin{equation}
    \mathcal{B}'\leqslant \max_{\rho}\sum_a \max_x \left\{ p(x) \mathrm{Tr}(\rho M^a_x) \right\},
\end{equation}
and thus, $\mathcal{B}'\leqslant \mathcal{D}$ [cf. Eq. (\ref{Drho})]. Furthermore, the upper bound of $\B'$ in the above equation is achieved within the entanglement-assisted scenario by considering Bob's measurements to be a classical post-processing on the outcome $a$, and hence, $\B' = \D$. This completes the proof.
\end{proof}
Hence, a steerable state provides an advantage in the measurement distinguishability task over the single-system scenario. In fact, any set of measurements $M_x$ gives rise to a steering inequality $\mathcal{B}\leqslant \mathcal{D}$ whose violation indicates that the entanglement-assisted scenario is advantageous over the single-system one in the measurement distinguishing task.

\section{Advantage of the entanglement-assisted scenario over the single-system one}\label{advantage}

In this article, we mostly restrict ourselves to instances where the measurements are perfectly distinguishable with the maximally entangled state, and they are drawn from a uniform ensemble. In the following, we discuss the advantage of using quantum entanglement in a measurement discrimination problem for different measurement scenarios.

\subsection{Advantage for rank-one projective measurements}

\begin{thm}\label{advantage projective}
A set of $d$ distinct rank-one projective measurements in dimension $d$ defined by the vectors $\{|v_x^a\ra \}$ where $x,a \in \{0,\dots,d-1\}$ ($x,a$ denote measurement setting and outcome respectively), is perfectly distinguishable in entanglement-assisted scenario but not perfectly distinguishable with single system, whenever the vectors satisfy the following relations,
\begin{enumerate}
    \item $\forall x,x',a, \ \la v_x^a|v^a_{x'}\ra = \delta_{x,x'}$ \ ,
    \item there exists $a,a'$ such that $|\la v^a_x|v^{a'}_{x'} \ra| < 1 $ for all $x,x'$.
\end{enumerate}
\end{thm}
\begin{proof}We show that in the entanglement-assisted scenario there exists a quantum strategy (Bob's measurements along with an entangled state) achieving one
in Eq. \eqref{qbxs}, whereas in the single-system scenario for any quantum strategy the value of (\ref{D}) does not reach one.

In the entanglement-assisted scenario let us consider that Alice and Bob share the maximally entangled state in $\mathbbm{C}^d\otimes\mathbbm{C}^d$,
\begin{equation}\label{maxent}
|\phi^+\ra=\frac{1}{\sqrt{d}}\sum_{i=0}^{d-1}|ii\rangle
\end{equation}
and that Bob's measurements are defined by the following measurement operators
\be 
N_a^x = (|v_x^a\ra\!\la v^a_x|)^T,
\ee 
where $T$ stands for the transposition in the standard basis.
Taking into account the first condition and the fact that transposition respects all the properties of a projector, we have $N_a^xN_a^{x'} = \delta_{x,x'}N_a^x$ for all $a,x,x'$ as well as $\sum_x N_a^x = \I $.

Let us now recall the following property of the maximally entangled state. 
\begin{fact}\label{fact:phi+}
For any two operators $A,B$ acting on $\mathbbm{C}^d$, $A\otimes B|\phi^+\ra = \I\otimes BA^T|\phi^+\ra$.
\end{fact}
Using it and taking $p(x)=1/d$ in the expression of $\mathcal{B}$ in Eq. \eqref{qbxs} we obtain,
\bea 
\hspace{-1cm}\B^{|\phi^+\rangle} &=& \frac{1}{d} \sum_{x,a} \la \phi^+ | |v^a_{x}\ra\!\la v^a_x| \otimes (|v^a_{x}\ra\!\la v^a_x|)^T |\phi^+\ra \nonumber \\
&=& \frac{1}{d} \sum_{a} \la \phi^+ | \I \otimes \underbrace{\bigg( \sum_x|v^a_{x}\ra\!\la v^a_x| \bigg) ^T}_{\I} |\phi^+\ra = 1.
\eea
On the other hand, the distinguishability given in Eq. \eqref{D} in this case can be stated as
\be \label{D projective}
\D = \frac{1}{d}\max_{|\psi\ra} \sum_a \max_x |\la \psi|v_x^a\ra |^2, \ee 
and it amounts to one if, and only if there exists $\ket{\psi}$ such that the following conditions
\be 
\max_x |\la \psi|v_x^a\ra |^2 = 1 
\ee 
hold true with $a=0,\ldots,d-1$. Imagine then that there exists such a state $\ket{\psi}$ that the above conditions are satisfied. Then, it is not difficult to see that for any pair $a,a'$ such that $a\neq a'$ there exist $x$ and $x'$ such that 
\begin{equation}
    |\psi\ra = |v^a_x\ra = |v^{a'}_{x'}\ra,
\end{equation}
which certainly contradicts assumption 2 of the theorem.
This completes the proof. 
\end{proof}

\subsubsection{Example}
 
As we show in Appendix \ref{AppA} in the simplest cases $n=d=2,3$ there are no projective measurements satisfying the two conditions of Theorem \ref{advantage projective}, and the first nontrivial case for which such measurements can be constructed is $n=d=4$. An exemplary choice of the corresponding vectors $\ket{v_x^a}$ is presented in the table below.
\begin{widetext}
\begin{center}
 \begin{table}[h]
 \label{TabI}
\begin{tabular}{|c|c|c|c|c|}
\hline
& $a=0$ & $a=1$  & $a=2$  & $a=3$  \\ \hline
$x=0$ & $|0\ra$  & $|1\ra$ & $|2\ra$ & $|3\ra$ \\ \hline 
$x=1$ & $(|1\ra-|2\ra)\sqrt{2}$  & $(|0\ra+|3\ra)\sqrt{2}$ & $(|0\ra-|3\ra)\sqrt{2}$ & $(|1\ra+|2\ra)\sqrt{2}$  \\ \hline
$x=2$ & $(|1\ra+|2\ra+|3\ra)/\sqrt{3}$  & $(|0\ra+|2\ra-|3\ra)/\sqrt{3}$ & $(|1\ra-|3\ra-|0\ra)/\sqrt{3}$ & $(|1\ra-|2\ra+|0\ra)/\sqrt{3}$  \\ \hline
$x=3$ & $(|1\ra+|2\ra-2|3\ra)\sqrt{6}$  & $(-|0\ra+2|2\ra+|3\ra)\sqrt{6}$ & $(2|1\ra+|3\ra+|0\ra)\sqrt{6}$ & $(|1\ra+|2\ra-2|0\ra)\sqrt{6}$  \\ \hline
\end{tabular}
\caption{An example of a list of the vectors $|v_x^a\ra$ in dimension $4$ satisfying the two sets of relations given in Theorem \ref{advantage projective}.}
\end{table}
\end{center}
\end{widetext}
One can verify that the above set of vectors satisfy the conditions given in Theorem \ref{advantage projective} and hence $\D<1$.

\subsubsection{Numerical optimization to find the value of $\D$ } \label{sec:nu}
In order to find the maximum value of $\D$, we consider the following parameterization of a pure state in $\mathbbm{C}^d$ \cite{bengtsson}
\begin{eqnarray}
    |\psi\ra&=&\cos\theta_1 |0\ra + \sum^{d-2}_{k=1} \left[ \left( \prod^{k}_{i=1}  \sin\theta_i \right) \cos\theta_{k+1} \  \mbox{e}^{\mathbbm{i}\nu_{k}} |k\ra \right] \nonumber\\
    && + \left( \prod^{d-1}_{i=1}  \sin\theta_i \right) \mbox{e}^{\mathbbm{i}\nu_{d-1}} |d-1\ra,
\end{eqnarray}
where $\theta_i\in [0,\pi/2]$ and $\nu_i\in [0,2\pi]$ (for $i=1,\dots,d-1$) are $2(d-1)$ number of unknown parameters. Given any set of measurements $\{M_x\}$, we can obtain the value of $\D$ by performing a numerical optimization of the expression \eqref{D} over the parameters, $\theta_i$ and $\nu_i$. For the exemplary set of four-dimensional measurements given in Table \ref{TabI}, we find $\D\approx 0.7752$ after carrying out such optimization for Eq. (\ref{D projective}). Since, there are only six parameters in this case, this value is expected to be the global maximum.
Therefore, we have a clear advantage in distinguishing the projective measurements with the maximally entangled state. We follow this procedure to find the values of $\D$ in the other examples given later.

In the following subsections we discuss the advantage of entanglement-assisted scenario for more general measurements or POVMs.

\subsection{Advantage for $\boldsymbol{d^2-}$outcome POVMs}

Given a dimension $d$, let us denote by $\w = \exp(2\pi \mathbbm{i}/d)$ the $d$-th root of unity, and define the following unitary matrices
\be \label{ZX}
Z = \sum^{d-1}_{i=0} \w^{i}|i\ra\!\la i|, \quad X = \sum^{d-1}_{i=0} |i+ 1\ra\!\la i|,
\ee 
where the sum($+$) is taken to be modulo $d$ sum. Let us then introduce the following set of $d^2$ unitary matrices,
\be \label{Ukl}
U_{k,l} = X^kZ^l,
\ee 
where $k,l = 0,\dots,d-1$. 

We can now state one of our main results. 
\begin{thm} \label{result:dd}
Given any set of orthonormal vectors $\{|v_i\ra\}^{d-1}_{i=0}$ such that for at least one pair $k,l$,
\be \label{cond}
\forall i,j, \quad | \la v_j|U_{k,l}|v_i\ra | < 1,
\ee 
there exists a set of $d$ POVMs, each having $d^2$ outcomes, that are perfectly distinguishable in entanglement-assisted scenario but not perfectly distinguishable in the single-system scenario. 
\end{thm}
\begin{proof}
Consider the following positive semi-definite 
rank-one operators acting on $\mathbbm{C}^d$, 
%
%$d$ POVMs defined by the following measurement operators
%
\be \label{Mkl}
M^{k,l}_x = \frac{1}{d} \ U_{k,l} |v_x\ra\!\la v_x| U_{k,l}^\dagger 
\ee 
where $x=0,\dots,d-1$ and $a = (k,l)$ denote the measurement settings and the outcomes, respectively. Using the following relation (see, e.g., Ref. \cite{Dariano}) 
\be 
\sum^{d-1}_{k,l=0} U_{k,l} \ \Xi \ U_{k,l}^\dagger = d \ \tr(\Xi) \I ,
\ee 
that holds true for any operator $\Xi$ acting on $\mathbbm{C}^d$, 
one easily finds that 
\begin{equation}
\sum_{k,l=0}^{d-1} M^{k,l}_x = \I,
\end{equation}
and thus the $d^2$ operators $M^{k,l}_x$ with $k,l=0,\ldots,d-1$ form a valid POVM for each $x$. 

Our aim now is to show that the above generalized measurements are perfectly distinguishable in the entanglement-assisted scenario, while not in the single-system scenario. To show that $\B =1$, let Alice and Bob share the maximally entangled state of two qudits $|\phi^+\ra$ (\ref{maxent}). Moreover, let Bob's measurements be given by
\be \label{Nkl}
N^x_{k,l} = (U_{k,l} |v_x\ra\!\la v_x| U_{k,l}^\dagger )^T = d (M^{k,l}_x)^T .
\ee 
It is trivial to see that $N^x_{k,l}N^{x'}_{k,l} = \delta_{x,x'} N^x_{k,l}$ for all $k,l$ since $|v_x\ra$ are orthogonal. Further, for any pair $k,l$,
\be 
\sum_{x} N^x_{k,l} = \left(U_{k,l} \sum_x |v_x\ra\!\la v_x| U_{k,l}^\dagger\right)^T = \I,
\ee 
thereby, for any pair $k,l$, the operators $N^x_{k,l}$ form a valid projective measurement. Let us then calculate the following quantity taking $p(x)=1/d$,
\bea
\B^{|\phi^+\rangle} &=& \frac{1}{d} \sum_{x,(k,l)} p((k,l),x|x,(k,l)) \nonumber \\
&=& \frac{1}{d} \sum_{x,k,l} \la \phi^+ | M^{k,l}_{x} \otimes N^{x}_{k,l} |\phi^+\ra \nonumber \\
&=& \frac{1}{d^2} \sum_{x,k,l} \la \phi^+ | (N^{x}_{k,l})^T \otimes N^{x}_{k,l} |\phi^+\ra,
\eea 
where we have used the relation in \eqref{Nkl}. Using Fact \ref{fact:phi+} the above expression can be further simplified as,
\bea 
\B^{| \phi^+\rangle} &=& \frac{1}{d^2} \sum_{x,k,l} \la \phi^+ | (\I \otimes |v_{x}\ra\!\la v_x|) |\phi^+\ra \nonumber \\
&=& \frac{1}{d^2} \sum_{k,l} \la \phi^+ | (\I \otimes \underbrace{\sum_x|v_{x}\ra\!\la v_x|}_{\I}) |\phi^+\ra = 1 .
\eea 
On the other hand, to show that $\D$ is strictly less than 1, 
let us first see that 
\begin{equation}
\la \psi| M^{k,l}_x|\psi\ra = \frac{1}{d} | \la\psi|U_{k,l}|v_x\ra|^2 
\end{equation}
for an arbitrary pure state $|\psi\ra\in\mathbbm{C}^d$. Subsequently, we have 
\be 
\D = \frac{1}{d^2}\max_{|\psi\ra} \sum_{k,l} \max_{x} \left\{ | \la\psi|U_{k,l}|v_x\ra|^2 \right\} .
\ee 
The above expression is 1, if and only if there exists a state $\ket{\psi}$
such that 
\be \label{D1cond}
\max_x\left\{ | \la\psi|U_{k,l}|v_x\ra|^2 \right\} =1
\ee 
is satisfied for any pair $k,l$.
For $k=l=0$, $U_{00}=\mathbbm{1}$ and therefore Eq. \eqref{D1cond} implies $|\psi\ra = |v_{x'}\ra$ (up to a phase)
for some $x'$. So, we can replace $|\psi\ra$ by $|v_{x'}\ra$ in Eq. \eqref{D1cond} for the other pairs $k,l$. Note that Eq. \eqref{cond} does not hold for $k=l=0$. So, there exists another pair $k,l$ for which \eqref{cond} must hold, and for that $k,l,$ Eq. \eqref{D1cond} is not satisfied. Hence, $\D < 1$ which completes the proof. 
\end{proof}

\subsubsection{Existence of $|v_x\ra$ in all dimension}

Let us now explore whether one can find an orthonormal basis $\ket{v_i}$ $(i=0,\ldots,d-1)$ in $\mathbbm{C}^d$ for any $d\geqslant 2$ that satisfies 
the condition (\ref{cond}).
To this end, we prove the following fact.

%\textbf{(Remik: I would be careful.)}
%
\begin{proposition} The eigenvectors of any unitary $U$ satisfy the condition \eqref{cond}, whenever $U$ and $X^kZ^l$ do not share any common eigenvector for all $k,l,$ except $k=l=0$.\end{proposition}
\begin{proof}
We prove this statement by contradiction.
Say $|v_i\ra $ are the eigenvectors of $U$. 
Negation of \eqref{cond} states that for all $k,l$ there exists $i,j$ such that 
\be \label{facteq1}
X^kZ^l|v_i\ra = \mbox{e}^{\mathbbm{i}\theta_{k,l}}|v_j\ra 
\ee 
for some $\theta_{k,l}\in\mathbbm{R}$. Since they do not share any common eigenvector, $i$ must be different than $j$ for all $k,l$ such that $k\neq 0$ or $l\neq 0$. Note that there are $d^2-1$ operators $X^kZ^l$ for which \eqref{facteq1} holds. On the other hand, there are at most $d$ distinct eigenvectors of $U$ which means there are at most ${d \choose 2} = d(d-1)/2$ possibilities to choose two distinct vectors among them. Thus, for Eq. \eqref{facteq1} to hold where $i\neq j$, there are at least two different operators, say $X^kZ^l$ and $X^{k'}Z^{l'}$, such that there exist the same pair $i,j$ for which
\be 
X^kZ^l|v_i\ra = \mbox{e}^{\mathbbm{i}\theta}|v_j\ra, \  X^{k'}Z^{l'}|v_i\ra = \mbox{e}^{\mathbbm{i}\theta'}|v_j\ra ,
\ee 
for some $\theta,\theta'\in\mathbbm{R}$.
Taking the conjugate transpose of the first one and multiplying with the second, we get,
\be 
| \la v_i|X^{k'-k}Z^{l'-l} |v_i\ra | = 1 .
\ee 
Now we arrive at a contradiction since the above implies $|v_i\ra$ is an eigenvector of $X^{k'-k}Z^{l'-l}$. 
\end{proof}
Many examples of such unitary can be found in all dimension. Let us mention one of those,
\be 
U_d=\sum_{i=0}^{d-1}\omega^{i+\frac{1}{2}}|i\rangle\!\langle i| -\frac{2}{d}\sum_{i,j=0}^{d-1}(-1)^{\delta_{i,0}+\delta_{j,0}}\omega^{\frac{i+j+1}{2}}|i\rangle\!\langle j| ,
\ee 
that has been discussed in \cite{sarkar2019selftesting}.

Let us then present another example in Hilbert spaces of dimension $d=2^r$ for any $r\geqslant 2$. To this end, consider the following states 
\begin{equation}\label{magic state}
|v_0\rangle=\cos\beta |0\rangle+\mbox{e}^{\mathbbm{i} \alpha}\sin\beta|1\rangle\, \end{equation}
and
\begin{equation}
|v_1\rangle=\mbox{e}^{-\mathbbm{i} \alpha}\sin\beta |0\rangle-\cos\beta|1\rangle,
\end{equation}
where $\beta=(0,\pi/4)$. These two states form a basis in $\mathbbm{C}^2$. Subsequently, by taking their tensor products we can create a basis in $\mathbbm{C}^d$ with $d=2^r$ in the following way
\be 
|v_{x_1\ldots x_r}\rangle = \bigotimes^r_{i=1} |v_{x_i}\rangle,
\ee 
where $x_1\cdots x_r=0,1$.
For instance, the vectors $|v_0\rangle\otimes|v_0\rangle$, $|v_0\rangle\otimes|v_1\rangle$,  $|v_1\rangle\otimes|v_0\rangle$ and  $|v_1\rangle\otimes|v_1\rangle$ form a basis in the two-qubit Hilbert space. 

Now, it is easy to check that $U_{0,1}=Z$ can be decomposed as 
as the following tensor product
\begin{eqnarray}\label{Z factorize}
Z=\bigotimes_{i=1}^{r}V_i,
%
%\sigma_z^{1/2^i}.
\end{eqnarray}
where $V_i$ is a $2\times 2$ unitary matrix of the form $V_i=\mathrm{diag}(1,\omega^{2^{r-i-1}})$. 
One can easily verify that for any $k$, $|\langle v_i |V_k|v_j\rangle|<1$ with $i,j=0,1$. Therefore 
\begin{equation}
 | \la v_{x_1\ldots x_r}|Z|v_{y_1\ldots y_r}\ra | < 1   
\end{equation}
for any configuration of $x_1,\ldots,x_r,y_1,\ldots,y_r=0,1$.
Hence, the entanglement-assisted scenario provides advantage over a single system. Moreover, using Eq. (\ref{Z factorize}) a similar conclusion can be drawn for the powers of $Z$, which allows to further lower the value of $\D$. 
%
%in terms of the tensor product and one can check that $| \la v_i|Z^n|v_j\ra | < 1$, 
%

%\re
%\subsection{Unbounded advantage with increasing dimension} 
%Things to do: Find $|v_x\ra$ such that $\D$ decreases with $d$, or $\frac{\B}{\D}$ %can be arbitrarily large. \\
%\blk  

\subsection{Advantage for informationally complete POVMs}

Interestingly, the POVMs defined in \eqref{Mkl} are informationally complete (IC) when the orthogonal vectors satisfy the following condition \cite{Dariano},
\be \label{icpovmcond}
\forall i,k,l,\quad | \la v_i|U_{k,l}|v_i\ra |\neq 0.
\ee 

\begin{proposition}
The set of IC-POVMs constructed from the orthogonal vectors stated in \eqref{Mkl} are also perfectly distinguishable in entanglement-assisted scenario but not perfectly distinguishable with single system. 
\end{proposition}

\begin{proof}
We need to prove the above condition \eqref{icpovmcond} implies \eqref{cond}. The contrapositive statement is more obvious. Negation of \eqref{cond} states that for all $k,l$ there exists $i,j$ such that 
\be \label{Uves}
U_{k,l}|v_i\ra = \mbox{e}^{\mathbbm{i}\theta_{k,l}}|v_j\ra 
\ee 
for some $\theta$ depending on $k,l$. Since $U_{1,0}=X,U_{0,1}=Z$ do not share any single common eigenstate, for $X$ or $Z$, \eqref{Uves} can hold only when $i\neq j$. Then multiplying both side of \eqref{Uves} by $\la v_i|$, we obtain a contradiction of \eqref{icpovmcond} either for $X$ or $Z$.
\end{proof}

\subsubsection{Examples of $|v_x\ra$}

For $d=2$, the vectors given in Eq. (\ref{magic state}) with an extra condition $\alpha=(0,\pi/2)$, satisfy the condition \eqref{icpovmcond} for IC-POVM. Taking $\alpha=\pi/4$ and $\beta=\cos^{-1}({1/\sqrt{3}})/2$, one obtains $\D\approx 0.7887$ for the respective IC-POVMs. 

For $d=3$, consider the following unnormalized basis 
\bea \label{icpovmd3}
|v_0\ra&=&|1\ra-|2\ra,\nonumber\\
|v_1\ra&=&(1+\sqrt{3})|0\ra+|1\ra+|2\ra,\,\nonumber\\
|v_2\ra&=&(1-\sqrt{3})|0\ra+|1\ra+|2\ra.
\eea
This basis satisfies the condition \eqref{icpovmcond} to be IC-POVM. By performing a simple optimization over four parameters (see Sec. \ref{sec:nu}), we get $\D\approx0.6436$ for the IC-POVMs obtained from the above vectors in \eqref{icpovmd3}. 

For $d=4$, we consider the tensor product basis $|v_0\rangle\otimes|v_0\rangle$, $|v_0\rangle\otimes|v_1\rangle$,  $|v_1\rangle\otimes|v_0\rangle$ and  $|v_1\rangle\otimes|v_1\rangle$, where $|v_0\ra$ and $|v_1\ra$ are given in Eq. (\ref{magic state}) with $\alpha=\pi/4$ and $\beta=\cos^{-1}({1/\sqrt{3}})/2$. From an optimization over six parameters described in Sec. \ref{sec:nu}, we find $\D\approx 0.622$ for the respective IC-POVMs.

For $d=2^r$ dimension, consider the tensor product basis of the vectors given in Eq. (\ref{magic state}) with $\alpha=\pi/4$ and $\beta=\cos^{-1}({1/\sqrt{3}})/2$. We verify that the vectors constructed this way satisfy the condition of IC-POVM given in Eq. (\ref{icpovmcond}) for $r=1,2,3,4,5$. Possibly the vectors may satisfy the condition for any $d=2^r$.

\subsection{Advantage for $\boldsymbol{(d+1)-}$outcome POVMs}

In the above scenario, Bob requires to perform $d^2$ measurements that grows polynomially with $d$. We now show another result of similar kind, where the number of measurements on Bob's side (or the outcome of the POVMs that we want to distinguish) is $d+1$.

\begin{thm}\label{result:d+1}
Given any set of orthogonal vectors $\{|v_i\ra\}^{d-1}_{i=0}$ that satisfy the conditions,
\bea \label{condd+1}
|\la j|v_i\ra | = 
\frac{1}{d}\begin{cases}
1, \qquad & i=j, \\[1ex]
\sqrt{d+1}, &   i\neq j ,
\end{cases}
\eea 
there exists a set of $d$ POVMs, each having $d+1$ outcomes, that are perfectly distinguishable in entanglement-assisted scenario but not perfectly distinguishable with single system. Here $|j\rangle$ is the computational basis. 
\end{thm}
\begin{proof}
Let us begin by using the orthonormal vectors 
$\ket{v_i}$ to introduce the following unitary matrix
\be \label{Uvx}
U = \sum_{i=0}^{d-1} |v_i\ra\!\la i|.
\ee
Thus, $U$ is the unitary that takes the computational basis to the basis $\{|v_i\ra\}$.

Let us then define the following vectors, 
\begin{equation} \label{etad}
|\eta^a_{x}\ra  = Z^aU|x\ra 
\end{equation}
for $a=0,\dots,d-1$ and
\begin{equation}
|\eta^{d}_x\ra = |x\ra
\end{equation}
for $a=d$, where $x=0,\ldots,d-1$ and $Z$ is given in Eq. \eqref{ZX}.

%for the setting index $x = 0,\dots,d-1$ and outcome index $a$.
%\bea \label{etad}
%&&\text{For }a=0,\dots,d-1, \ |\eta^a_{x}\ra  = Z^aU|x\ra ,\mbox{ and}\nonumber\\
%&&\mbox{for } a=d, |\eta^{d}_x\ra = |x\ra,  
%\eea 
%where $Z$ is given in \eqref{ZX} and 

With the aid of these vectors we consider the POVM elements,
\be \label{d+1 outcome POVM}
M^a_x = \frac{d}{d+1} |\eta^a_x\ra\!\la\eta^a_x| .
\ee 
Clearly, the eigenvalues of these operators are non-negative and less than 1. Let us then show that $\sum_aM_x^a=\I$. To do so, we need the following relation that for any operator $\Xi = \sum_{i,j} c_{i,j}|i\ra\!\la j|$ acting on $\mathbbm{C}^d$,
\bea  \label{za}
\sum^{d-1}_{a=0} Z^a \ \Xi \ (Z^a)^\dagger &=& \sum_{i,j} \sum^{d-1}_{a=0} \w^{a(i-j)} c_{i,j}|i\ra\!\la j| \nonumber \\
&=& d \sum_i c_{i,i} |i\ra\!\la i| ,
\eea  
where we have applied the following identity, 
\be 
\sum^{d-1}_{a=0} \w^{ak} = 
\begin{cases}
 0, \quad k\neq 0 \\
 d, \quad k= 0.
\end{cases}
\ee 
Replacing $|\eta^a_x\ra$ from \eqref{etad} and using the above relation \eqref{za}, we find
\bea 
\sum_{a=0}^{d-1} M_x^a &=& \frac{d}{d+1} \sum\limits_{a=0}^{d-1} \ Z^a |v_x\ra\!\la v_x|(Z^a)^\dagger + \frac{d}{d+1} |x\ra\!\la x|  \  \nonumber \\
&=& \frac{d}{d+1}\Bigg[ d\sum_{i\neq x}  |\la i|v_x \ra|^2 |i\ra\!\la i| \nonumber\\
&&\qquad\qquad\, + \left(d |\la x|v_x \ra|^2 + 1 \right) |x\ra\!\la x|\Bigg].
\eea  
Due to the conditions \eqref{condd+1}, the above quantity is $\I$ and hence, $M_x^a$ form a valid POVM for all $x$.  

To show $\B =1$, let Alice and Bob share maximally entangled state (\ref{maxent}), and let Bob's measurements be given by
\be \label{Nkl1}
N^x_{a} = (|\eta^a_x\ra\!\la\eta^a_x|)^T = \frac{d+1}{d} (M^a_x)^T .
\ee 
It is trivial to see that $N^x_{a}N^{x'}_{a} = \delta_{x,x'} N^x_{a}$ for all $a$, and
$\sum_{x} N^x_{a} = \I$,
thereby $N^x_{a}$ form valid projective measurements. Following the similar method as in the proof of previous theorems we can readily show $\B =1$, by taking $p(x)=1/d$. 

Now to check that $\D<1$, we first note that
\be 
\la \psi| M^{a}_x|\psi\ra = 
\frac{d}{d+1}\begin{cases}
 | \la\psi|Z^a|v_x\ra|^2,  &  a=0,\dots,d-1 \\[1ex]
| \la\psi|x\ra|^2, &  a=d,
\end{cases}
\ee 
where $|\psi\ra$ is an arbitrary pure state from $\mathbbm{C}^d$. Subsequently, we have 
\bea 
\D &=& \frac{1}{d+1} \max_{|\psi\ra} \Bigg[ \sum^{d-1}_{a=0} \max_{x} \{ | \la\psi|Z^a|v_x\ra|^2 \} \nonumber\\
&&\qquad\qquad\qquad\quad + \max_x \{ | \la\psi|x\ra|^2 \}\Bigg] .
\eea 
The above expression is 1, if and only if, for all $a=0,\dots,d-1$,
\begin{equation} \label{D1cond2}
\max_x\{ | \la\psi|Z^a|v_x\ra|^2 \} =1
\end{equation}
and
\begin{equation}
\max_x\{ | \la\psi|x\ra|^2 \} =1.
\end{equation} 
For $a=0$, Eq. \eqref{D1cond2} implies $|\psi\ra = |v_x\ra$ for some $x$. In that case, $\max_x\{ | \la v_x |x\ra|^2 \}$ cannot be one due to Eq. \eqref{condd+1}. Hence, $\mathcal{D} < 1$ which completes the proof. 
\end{proof}

\subsubsection{Examples in $d=2,3,4$}\label{sec:exp234}

Let us here present exemplary sets of vectors $\ket{v_x}$ satisfying
(\ref{condd+1}) for $d=2,3,4$.

For $d=2$, we consider the following two orthogonal vectors 
\be \label{v0v1d2}
|v_0\ra = \frac{1}{2}\left(|0\ra+\sqrt{3}|1\ra\right), \quad |v_1\ra=\frac{1}{2}\left(\sqrt{3}|0\ra-|1\ra\right),
\ee
which satisfy the conditions given in Eq. (\ref{condd+1}). The optimal value of $\D$ for the measurements defined by \eqref{etad}-\eqref{d+1 outcome POVM} is $5/6$ (see the next subsection for a proof). This particular measurement set was proposed by  Sedl\'ak and Ziman  in Ref. \cite{ziman}.

For $d=3$, consider the following set of orthogonal vectors, 
\be
|v_i\ra = \frac{1}{3} |i\ra - \frac{2}{3} |i+1\ra - \frac{2}{3} |i+2\ra ,
\ee
where $i=0,1,2$ and the sum inside $|\cdot \ra$ is modulo $3$. They readily satisfy conditions given in Eq. \eqref{condd+1}. By performing an optimization over four parameters described in Sec. \ref{sec:nu}, we find $\D \approx 0.698$ for the POVMs defined by \eqref{etad}-\eqref{d+1 outcome POVM}.

For $d=4$, the vectors are as follows,
\bea
 |v_0\ra &=& \frac{1}{4} |0\ra +  \frac{\sqrt{5}}{4}  |1\ra + \frac{\sqrt{5}}{4} |2\ra + \frac{\sqrt{5}}{4} |3\ra  \nonumber \\
|v_1\ra &=& \frac{\sqrt{5}}{4} |0\ra -  \frac{1}{4} |1\ra + \frac{\sqrt{5}}{4} |2\ra - \frac{\sqrt{5}}{4} |3\ra  \nonumber \\
|v_2\ra &=& \frac{\sqrt{5}}{4} |0\ra -  \frac{\sqrt{5}}{4} |1\ra - \frac{1}{4} |2\ra + \frac{\sqrt{5}}{4} |3\ra  \nonumber \\
|v_3\ra &=& \frac{\sqrt{5}}{4} |0\ra +  \frac{\sqrt{5}}{4} |1\ra - \frac{\sqrt{5}}{4} |2\ra - \frac{1}{4} |3\ra .
\eea 
Again, a simple optimization over six parameters yields $\D \approx 0.706$ for the respective POVMs.

\subsection{All pure two-qubit entangled states provide advantage}\label{steerability two-qubit}

So far, we have employed only the maximally entangled state \eqref{maxent} for presenting the merit of using entanglement-assisted scenario. Let us probe whether such merit persists if we use non-maximally entangled states. 
For this purpose, we consider the pair of two-dimensional POVMs $M_x=\{M_x^a\}_{a=0}^2$ $(x=0,1)$ introduced in Sec. \ref{sec:exp234} as follows
\begin{eqnarray}
%\label{TintodeVerano}
&&\hspace{-0.5cm}M_i^{0}=\frac{2}{3}|v_i\ra\la v_i|,\;\;M_i^{1}=\frac{2}{3}|\tilde{v}_i\ra\!\la\tilde{v}_i|,\;\; M_i^{2}=\frac{2}{3}|i\ra\!\la i|, \label{trine1}
\end{eqnarray}
where $\ket{v_i}$ $(i=0,1)$ are defined in 
Eq. (\ref{v0v1d2}) and $\ket{\tilde{v}_i}=Z\ket{v_i}$. 
This pair is perfectly distinguishable using the two-qubit maximally entangled state \cite{ziman}. 
Here, we aim to show that if Alice and Bob share any pure two-qubit entangled state, then the same pair of measurements (up to unitary rotation) can be distinguished with higher probability than the optimal probability obtained with single systems. To this end, we state the following theorem.
\begin{thm}\label{advantage two qubit}
For any two-qubit pure entangled state $|\phi\ra$, there exists a set of two three-outcome POVMs such that $\B^{\ \phi} > \D$. 
\end{thm}
\begin{proof}
Consider the two POVM defined
by the measurement operators $M_x^a$ given in Eq. (\ref{trine1}),
%\begin{eqnarray}
%&&M_0^{0}=\frac{2}{3}|v_0\ra\la %v_0|,\;M_0^{1}=\frac{2}{3}|\tilde{v}_0\ra\la\tilde{v}_0|, %M_0^{2}=\frac{2}{3}|0\ra\la0|; \label{trine1}\\
%&&M_1^{0}=\frac{2}{3}|v_1\ra\la %v_1|,\;M_1^{1}=\frac{2}{3}|\tilde{v}_1\ra\la\tilde{v}_1|, %M_1^{2}=\frac{2}{3}|1\ra\la1|, \label{trine2}
%\end{eqnarray}
where $x=\{0,1\}$ and $a=\{0,1,2\}$. Here $|v_0\ra, |v_1\ra$ are defined in Eq. \eqref{v0v1d2}, and $\ket{\tilde{v}_i}=Z\ket{v_i}$.
These measurement operators follow the construction given in Eq. (\ref{d+1 outcome POVM}) and since they all belong to the $x$-$z$ plane of the Bloch sphere, without loss of generality, we can consider the form of the state to be $|\psi\ra=\sin\delta |0\ra+\cos \delta |1\ra$, where $\delta\in [0,\pi/4]$. One can straightforwardly verify that for this state $\D$ expresses as
\begin{widetext}
\begin{eqnarray}
\D&=&\frac{1}{12} \Bigg\{4\max \{ \cos ^2\delta,\sin ^2\delta \}+\max \left\{ 2-\sqrt{3} \sin (2 \delta)-\cos (2 \delta), 2+\sqrt{3} \sin (2 \delta)+\cos (2 \delta)\right\}\nonumber\\
&&\hspace{1cm}+\max \left\{2-\sqrt{3} \sin (2 \delta)+\cos (2 \delta), 2+\sqrt{3} \sin (2 \delta)-\cos (2 \delta)\right\}\Bigg\}\nonumber\\
&=&\displaystyle\frac{1}{6} \times
\left\{\begin{array}{ll}
3+2\cos(2\delta), &0\leq\delta\leq\pi/12,\\[1ex]
3+\cos(2\delta)+\sqrt{3} \sin (2 \delta), &\pi/12\leq\delta\leq\pi/4
%
%2+ 2\cos ^2\delta+ \cos(2\delta),& 0\leq\delta\leq\pi/12\\[1ex]
%2+ 2 \cos ^2\delta+\sqrt{3} \sin (2 \delta),&\pi/12\leq\delta\leq\pi/4
\end{array}
\right..
\end{eqnarray}
\end{widetext}
It is easy to check that the optimal value of $\D$ is $5/6$ for $\delta=0$ and $\pi/6$. \\

For any pure two-qubit entangled state $|\phi\ra$, we know that there exists two local unitary operations $U_A$ and $U_B$ such that 
\be \label{2qs}
|\phi'\ra := U_A\otimes U_B |\phi\ra = \sin{\alpha} |00\ra + \cos{\alpha} |11\ra  
\ee 
for $\alpha \in (0,\pi/4]$. 
To find the value of $\B$, we first consider the non-maximally entangled state $|\phi^\prime\ra$. On Bob's side, we then consider three projective 
measurements  with the measurement operators given in the general form 
\begin{eqnarray}
N_a^{x}&=&\frac{1}{2}\mathbbm{1}+(-1)^x (\sin \theta_a  \cos \phi_a \sigma_x+\sin \theta_a \sin \phi_a  \sigma_y\nonumber\\
&&\hspace{2cm}+ \cos \theta_a \sigma_z),   
\end{eqnarray}
where $0\leq \theta_a\leq \pi$ and $0\leq \phi_a\leq 2\pi$. Using the form of $M^a_x$ given in the Eq. \eqref{trine1} and for the state  \eqref{2qs} we have
\begin{eqnarray}
\B^{\ \phi'} &=& \max_{\{N^x_a\}} \ \sum_{x,a} p(x) \langle \phi'|M^a_x \otimes N^{x}_{a}|\phi'\rangle\nonumber\\
&=&\frac{1}{12} \max_{\{\theta_0,\theta_1,\theta_2,\phi_0,\phi_1,\phi_2\}}\Big(6+2 \cos \theta_0-\cos \theta_1-\cos \theta_2 \nonumber \\ 
&& \hspace{2.8cm} - \sqrt{3} \sin (2\alpha) \sin \theta_2  \cos \phi_2 \nonumber\\
&&\hspace{2.8cm}  + \sqrt{3} \sin (2 \alpha)  \sin \theta_1 \cos \phi_1\Big)\nonumber\\
&=&\frac{1}{12}\max_{\{\theta_1,\theta_2\}}\Big[8+ \sqrt{3} \sin (2\alpha) (\sin \theta_1+ \sin \theta_2) \nonumber \\ 
 && \qquad \qquad \quad -\cos\theta_1-\cos\theta_2\Big].
\end{eqnarray}
Using the fact that $a\sin t-b\cos t\leq \sqrt{a^2+b^2}$, where the maximum is achievable for $\sin t=a/ \sqrt{a^2+b^2}$, we find
\begin{equation}\label{Bxnm}
\B^{\ \phi'}=\frac{1}{6}\Big(4+\sqrt{1+3\mathcal{C}^2}\Big),
\end{equation}
where $\mathcal{C}=\sin(2\alpha)$ is the concurrence \cite{rungta} of the state $|\phi^\prime\ra$. Hence, $\B>5/6$ as $\mathcal{C}>0$ whenever the state is entangled.  Remark that for the two-qubit maximally entangled state the value is one. The optimal observables for Bob are as follows
\be \label{Nax2}
N_0=\sin\theta \sigma_x-\cos\theta\sigma_z, \quad N_1=-\sin\theta \sigma_x-\cos\theta\sigma_z
\ee
and
\begin{equation}
    N_2=\sigma_z,
\end{equation}
where $\sin\theta=\sqrt{3}\mathcal{C}/\sqrt{1+3\mathcal{C}^2}$. 

Subsequently, for the general state $|\phi\ra$, the measurement set $U_AM^a_{x}U^\dagger_A$ along with the Bob's measurement $U_BN_y U^\dagger_B$ given in Eqs. \eqref{trine1} and Eq. \eqref{Nax2} respectively, achieve the same value of $\B^{\ \phi'}$ as given in Eq. \eqref{Bxnm}.
This completes the proof.
\end{proof}

We finally discuss the efficacy of a class of mixed states in the above mentioned measurement discrimination task. Let us consider the Werner state \cite{werner},
\begin{equation}\label{ws}
    \rho_W=p |\phi^+\ra\!\la \phi^+|+\frac{1-p}{4}\mathbbm{1},
\end{equation}
where $|\phi^+\ra$ is the two-qubit maximally entangled state given in Eq. \eqref{maxent}. The state \eqref{ws} is entangled for $p>1/3$ and steerable for $p>1/2$ \cite{wiseman}. If we use this state for discriminating the  measurements given in Eq. \eqref{trine1}, a straight forward calculation leads to
\begin{equation}
\B^{\rho_W}
%=\max_{N_a^x}\sum_{x,a}p(x)\mbox{Tr}\left(M_x^a\otimes N_a^x \rho_W\right)
=\frac{1+p}{2}.
\end{equation}
Thus, the Werner state \eqref{ws} provides advantage when $p>2/3$. Note that the state is steerable in the range $1/2<p\leqslant 2/3$, but does not show any advantage in this particular measurement distinguishability task. However, there may exists another measurement discrimination task for which this state provides advantage in that range.

\section{Conclusion and open problems} \label{conclusion}

In this article we discuss the single-shot measurement discrimination problem for an arbitrary $d$-dimensional quantum system. We introduce a framework allowing to study the problem in two scenarios: the single-system scenario and the entanglement-assisted one. Interestingly, entangled quantum systems can provide an advantage in a measurement distinguishability task. To be precise, we provide  criteria to ascertain whether a set of measurements (projective or generalized) can be perfectly discriminated with the aid of the maximally entangled state, but not in the single-system scenario. Furthermore, we prove that the advantage in the entanglement-assisted scenario is a witness of steerability of the underlying quantum state. Finally, we show that any pure two-qubit entangled state provides advantage in the measurement discrimination task. 
%\deleted{Therefore, in this paper, we find some interesting results in this direction. However, we believe the problem of discrimination of measurements is relatively new and less explored. Thus, in the following we discuss some possible future problems/directions. }

Let us also outline some possible directions for further research. 
In \cite{piani2005} and \cite{piani2015}, Piani and Watrous showed that entangled states and steerable states are necessary as well as sufficient for the advantage in channel discrimination and sub-channel discrimination task, respectively, in certain scenarios. Likewise, in section \ref{necessary steering}, we show that steerability is the necessary condition for the advantage in measurement distinguishability tasks. But, whether it is also a sufficient criteria or not, remains an open question. Since quantum measurements are reckoned as a sub-class of quantum channels, the following conjecture certainly enhances the previous results -- \emph{For every steerable state $\rho_{AB}$ by Bob, there exists a set of quantum measurements $\{M_x\}$ such that $\B^{\ \rho_{AB}} > \D$. Alongside, one may extend the result of Theorem \ref{advantage two qubit} to all the pure entangled states.} Furthermore, as maximally entangled states can perfectly discriminate the measurements given in theorems \ref{advantage projective}, \ref{result:dd} and \ref{result:d+1}, it would be interesting to look for a set of measurements in any $d$ such that the advantage is unbounded, or, to be precise, $\B/\D$ increases with $d$. In addition, one may investigate measurements for which the optimal probability is achieved by non-maximally entangled states. Besides, there has been some work on discriminating measurements without labelling the outcomes \cite{ziman2008, ziman2009} and hence, it would be interesting to generalize our protocol to unambiguous measurement discrimination problem.

\section*{acknowledgements} We are thankful to Anubhav Chaturvedi, Markus Grassl, Karol \.Zyczkowski, Micha\l{} Horodecki and Anindya Sengupta for helpful discussions.  C. D. acknowledges the support by the Foundation for Polish Science under the ``Quantum Optical
Technologies" project carried out within the International Research Agendas programme co-financed by the European Union under the European Regional Development Fund. T. B. acknowledges the support by the Foundation for Polish Science (IRAP project, ICTQT, contract no. 2018/MAB/5, co-financed by EU within Smart Growth Operational Programme). D. S. and R. A. acknowledge the support by Foundation for Polish Science through the First Team project (No First TEAM/2017-4/31).

\appendix

\section{A proof}
\label{AppA}

Here we show that for $n=d=2$ and $n=d=3$ there are no
orthonormal bases in $\mathbbm{C}^2$ and $\mathbbm{C}^3$, respectively which satisfy the assumptions 1 and 2 of Theorem \ref{advantage projective}. 

Let us begin with the simpler case of $n=d=2$ and consider two two-element orthonormal bases  in $\mathbbm{C}^2$
$\{|v_0^0\rangle,|v_0^1\rangle\}$ and $\{|v_1^0\rangle,|v_1^1\rangle\}$.
Without any loss of generality we can assume the first basis to be the
computational one, i.e., $|v_0^i\rangle=|i\rangle$. The assumption 1 of
Theorem \ref{advantage projective} imposes that also 
$\{|v_0^0\rangle,|v_1^0\rangle\}$ and $\{|v_0^1\rangle,|v_1^1\rangle\}$
are also orthonormal bases in $\mathbbm{C}^2$, implying that 
up to phases $|v_1^0\rangle=|1\rangle$ and
$|v_1^1\rangle=|0\rangle$. This means that also the second basis
$\{|v_1^0\rangle,|v_1^1\rangle\}$ is the computational one.
Then, however, the assumption 2 is violated. 

Let us now move on to the case $n=d=3$ and consider three orthonormal bases in $\mathbbm{C}^3$, $\{|v_x^0\rangle,|v_x^1\rangle,|v_x^2\rangle\}$ with $x=0,1,2$. As before we can assume that the first basis is the computational one, $|v_0^i\rangle=|i\rangle$. The first assumption of Theorem \ref{advantage projective} means that $\{|v_0^a\rangle,|v_1^a\rangle,|v_2^a\rangle\} $
with $a=0,1,2$ must also be orthonormal bases in $\mathbbm{C}^3$. It is not difficult to see that the most general form of the vectors $|v_x^a\rangle$ compatible with this requirement is given in the table below,
%
%\begin{widetext}
\begin{center}
 \begin{table}[h]
\begin{tabular}{c|c|c|c}
& $a=0$ & $a=1$  & $a=2$    \\ \hline 
$x=0$ & $|0\ra$  & $|1\ra$ & $|2\ra$  \\ 
$x=1$ & $a|1\ra+b|2\ra$  & $p|0\ra+q|2\ra$ & $s|0\ra-t|1\ra$  \\ 
$x=2$ & $b^*|1\ra-a^*|2\ra$  & $q^*|0\ra-p^*|2\ra$ & $t^*|0\ra-s^*|1\ra$  
\end{tabular}
%\caption{An example of the list of the vectors $|v_x^a\ra$ in dimension $4$.}
\end{table}
\end{center}
%\end{widetext}
%
where $a,b,p,q,s,t\in\mathbbm{C}$. We additionally need to impose that the vectors in the second and the third row of this table are pairwise orthogonal. It is, however, fairly easy to see that this last requirement can only be met if 
all the vectors $|v_x^a\rangle$ are actually elements of the computational basis of $\mathbbm{C}^3$. But then, as before, the second assumption of Theorem \ref{advantage projective} cannot be met.

\bibliography{ref} 

\end{document}